\documentclass[runningheads]{llncs}
\usepackage{graphicx}
\usepackage{epsfig}
\usepackage{graphics}
\usepackage{array}
\usepackage{amsmath}
\usepackage{algorithm}
\usepackage{epstopdf}
\usepackage{amssymb}
\usepackage{algcompatible}
\usepackage{fullpage}
\date{}

\title{r-connectivity Augmentation in Trees}
\author{S.Dhanalakshmi, N.Sadagopan and Nitin Vivek Bharti} 
\institute{Indian Institute of Information Technology, Design and Manufacturing, Kancheepuram, Chennai, India. \\
\email{$\{mat12d001, sadagopan\}@iiitdm.ac.in$}}
\begin{document}
\maketitle
\begin{abstract}


A \emph{vertex separator} of a connected graph $G$ is a set of vertices removing which will result in two or more connected components and a \emph{minimum vertex separator} is a set which contains the minimum number of such vertices, i.e., the cardinality of this set is least among all possible vertex separator sets. The cardinality of the minimum vertex separator refers to the connectivity of the graph G. A connected graph is said to be $k-connected$ if removing exactly $k$ vertices, $ k\geq 1$, from the graph, will result in two or more connected components and on removing any $(k-1)$ vertices, the graph is still connected. A \emph{connectivity augmentation} set is a set of edges which when augmented to a $k$-connected graph $G$ will increase the connectivity of $G$ by $r$, $r \geq 1$, making the graph $(k+r)$-$connected$ and a \emph{minimum connectivity augmentation} set is such a set which contains a minimum number of edges required to increase the connectivity by $r$. In this paper, we shall investigate a $r$-$connectivity$ augmentation in trees, $r \geq 2$. As part of lower bound study, we show that any minimum $r$-connectivity augmentation set in trees requires at least $ \lceil\frac{1}{2} \sum\limits_{i=1}^{r-1} (r-i) \times l_{i} \rceil $ edges, where $l_i$ is the number of vertices with degree $i$. Further, we shall present an algorithm that will augment a minimum number of edges to make a tree $(k+r)$-connected.\\ \\

\end{abstract}

\section{Introduction}
Connectivity augmentation is an optimization problem which finds applications in the study of computer networks, in particular, increasing the robustness of a network by adding a minimum number of links to the network \cite{weiner}. Connectivity augmentation problem asks for a minimum number of edges whose addition to a $k$-vertex connected graph results in a $(k+r)$-vertex connected graph, $k, r \geq 1$. The study of connectivity augmentation was initiated by Eswaran et al. \cite{eswaran} and it has attracted many researchers since then. Well known results in this area include biconnectivity augmentation of $1$-connected graphs \cite{eswaran,rosenthal,hsu,hsu1}, and tri-connectivity augmentation of biconnected graphs \cite{watanabe,hsu2}. A recent break-through result due to Vegh \cite{vegh}, shows that making a $k$-connected graph, a $(k+1)$-connected graph by augmenting a minimum number of edges is polynomial-time solvable.

Interestingly, parallel algorithm for biconnectivity augmentation has been looked at in the literature \cite{hsu,surabhi}. Further, connectivity augmentation in special graphs like chordal graphs has been studied in \cite{nsn}. Despite several attempts, $r$-connectivity augmentation in graphs for arbitrary is open. I.e., Given a $k$-connected graph $G$, find a minimum set of edges whose augmentation to $G$ makes it $(k+r)$-connected is open to the best of our knowledge.

In this paper, we investigate $r$-connectivity augmentation in trees. A natural approach is to use algorithm presented by Vegh \cite{vegh} iteratively $r$-times. However, this approach need not give optimum always. Therefore, this calls for a different structural understanding of graphs to determine a minimum connectivity augmentation set. Recently, in \cite{dhana}, Dhanalakshmi et al. initiated the study of $r$-connectivity augmentation in trees for the case $r=3$. In this paper, we look at this problem in a larger dimension and we have presented the lower bound for any $r$-vertex connectivity augmentation in trees. Also, we present an algorithm which takes a tree of order $n$ and an integer $r < n$ as an input and outputs the $r$-vertex connected graph by augmenting a minimum number of edges. We also believe that this result can be extended to $1$-connected graphs.\\

\noindent \textbf{Our Contribution:} We present the following results in this paper:
\begin{itemize}
\item Given a tree $T$, the number of edges to be augmented to convert a tree $T$ to a $r$-connected graph is at least $ \lceil\frac{1}{2} \sum\limits_{i=1}^{r-1} (r-i) \times l_{i} \rceil $, where $l_i$ denotes the number of vertices of degree $i$ in $T$.
\item A polynomial-time algorithm to determine the set of edges whose augmentation to the tree $T$ will make $T$, $r$-vertex connected graph meeting the above bound.
\end{itemize} 
\textbf{Organization of the paper:} We present the connectivity augmentation preliminaries in Section 2. The lower bound theory and an algorithm for $r$-connectivity augmentation is presented in Section 3. We conclude this paper with some directions for further research on $(k+r)$-connectivity augmentation in $k$-connected graphs.

\section{Preliminaries}
Notations used in this paper are as per \cite{nsn,dbwest}. Let \emph{G = (V,E)} represent an undirected connected graph where \emph{V(G)} denotes the non-empty set of vertices and $E(G)$ is an two element subset of $V(G)$. For any vertex ${v} \in V(G)$, $N_{G}(v) = \{u \mid \{u,v\} \in E(G) \}$ and $d_{G} = |N_{G}(v)|$ refers to the neighborhood and the degree of vertex $v$ in the graph $G$, respectively. ${ \delta(G) }$ and ${ \Delta(G) }$ are minimum and maximum degree of graph $G$, respectively. For simplicity, we use $\delta$ and $\Delta $ when the concerned graph is clear from the context.  A \emph{cycle} is a connected graph where $\delta = \Delta = 2$. A tree is a connected acyclic graph. A \emph{path} is a connected acyclic graph with $\Delta \leq 2$. For $S \subset V(G)$, $G[S]$ denotes the graph induced on set $S$ and $G\backslash S$ is the induced graph on the vertex set $V(G)\backslash S$. A vertex separator of a graph $G$ is a set $S \subseteq V(G)$ such that $G\backslash S$ has more than one connected component. A minimum vertex separator $S$ signifies a vertex separator of least size and the cardinality of such a set $S$ defines the vertex connectivity of the graph $G$, written as $\kappa(G)$. A graph is $k$-vertex connected if $\kappa(G)=k$. In particular, if $\kappa(G)=1$ then the graph is 1-connected i.e., the minimum vertex separator $S$ of $G$ is a singleton set, the vertex in $S$ is known as a \emph{cut-vertex} of $G$. For a graph \emph{G} with $\kappa (G)=k$, a minimum connectivity augmentation set $ E_{ca}  =  \{\{ u,v \}  |  u,v  \in V(G)\text{ and }\{ u,v \} \notin E(G)\}$ is such that the graph obtained from \emph{G} by augmenting $ E_{ca} $ edges is of vertex connectivity $k + r$, $r  \geq  2 $. 
\section{$r$-connectivity augmentation in trees}
In this section, we shall first present the lower bound analysis which gives the number of edges to be augmented to a tree in any minimum connectivity augmentation to make a tree $r$-vertex connected. Later, we give a sketch of the algorithm and then an algorithm with analysis which will output a connectivity augmentation set meeting the lower bound. Our approach finds a minimum $r$-connectivity augmentation set for paths and non-path trees, separately.

\begin{lemma}
\label{lb}
Let $T$ be a tree and $l_i$ denotes the number of vertices of degree $i$ in $T$. Then, any $r$-connectivity augmentation set $ E_{ca} $ is such that $| \: E_{ca} \: | \: \geq \:   \lceil \: \frac{1}{2} \sum\limits_{i=1}^{r-1} (r-i) \times l_{i} \: \rceil  $. 
\end{lemma} 
\begin{proof}
It is well-known that for any $r$-connected graph $ G, \: \delta(G) \geq r $. Therefore, to make \emph{T} a $r$-connected graph, we must increase the degree of each vertex with degree $i < r$ by at least $ r-i $. i.e., the sum of degrees to be increased is $ \lceil \: \sum\limits_{i=1}^{r-1} (r-i) \times l_{i} \: \rceil  $. Since, an edge joins a pair of vertices, any augmentation set $ E_{ca} $ has at least  $ \lceil \: \frac{1}{2} \sum\limits_{i=1}^{r-1} (r-i) \times l_{i} \: \rceil  $ edges. This proves the lemma. $\hfill \qed$
\end{proof}

\subsection{$r$-connectivity augmentation for paths}
We first present an algorithm for connectivity augmentation to convert a path, a tree with $\Delta \leq 2$, to a $r$-connected graph. \\

\noindent \textbf{Outline of the Algorithm:} The path is converted into a cycle by augmenting an edge between the two degree one vertices. If the required connectivity $r$ is odd, then we convert the 2-connected cycle into a 3-connected graph by augmenting edges in such a way that every edge creates a cycle of length $ \lfloor \frac{n}{2}  \rfloor \, +\, 1$ \cite{dhana}. After that, from every vertex $v_i$ of degree less than $r$ an edge is augmented from vertex $\: v_i \:$ to $ \: v_{i+j} \: $ where $j$ varies from 2 to $ \lfloor \frac{r}{2} \rfloor $ iteratively. This approach guarantees that the algorithm augments exactly $\lceil \: \frac{1}{2} \sum\limits_{i=1}^{r-1} (r-i) \times l_{i} \: \rceil$ edges and the resultant graph is $r$-connected.

\begin{algorithm}[H]
\caption{\tt $r$-connectivity augmentation of a Tree}
\begin{algorithmic}[1]
\STATE{\textbf{Input.} Tree $T$}
\STATE{\textbf{Output.} $r$-vertex connected graph $H$.}
\IF{$T$ is a path}
        \STATE{$Path\:Augmentation(T)$}
    \ELSE
    \STATE{ $Non-path\:Augmentation(T)$ }    
    \ENDIF
\STATE{Output $H$}
\end{algorithmic}
\label{alg:tree}
\end{algorithm}

\begin{algorithm}[H]
\caption{\tt $r$-connectivity Augmentation in path like trees: $Path \: Augmentation\,(Tree\:T)$}

\begin{algorithmic}[1]
\STATE{Let $P_n\:=\:(v_1,\,v_2,\,\ldots,\,v_n)  $ denotes an ordering of vertices of $T$ such that for all $ 1\: \leq i \leq n-1, \: v_i $ is adjacent to $ v_{i+1} $. }

\STATE{ Augment the edge $\{ v_1,\,v_n \}$ to $T$ and update $E_{ca}$. /* Converts a path into a cycle. */ }
\IF{$r$ is odd}
    \FOR{$i\,=\,1$ to $ \lceil \frac{n}{2} \rceil $}
            \STATE{Augment the edge $ \{ v_i,\, v_{\lfloor \frac{n}{2} + i \rfloor} \} $ to $T$ and update $E_{ca}$. /* Every augmented edge will create a $C_{\lceil \frac{n}{2} + i \rceil }$ */}
    \ENDFOR
\ENDIF
\FOR{$ j\,=\,2 $ to $ \lfloor \frac{r}{2} \rfloor $}
\STATE{ /* Each iteration will increase the connectivity of the graph by 2 */ }
\FOR{$ i\,=\,1$ to $ n $ }
    \IF{$ (i+j) \leq\, n$ }
    \STATE{Augment the edge $ \{ v_i,\, v_{i+j} \}$ to $T$ and
     update $ E_{ca} $}
    \ELSE
    \STATE{Augment the edge $ \{ v_i,\, v_{((i+j)\mod n)} \}$ to $T$ and
     update $ E_{ca} $}    
    \ENDIF
     \ENDFOR
\ENDFOR
\end{algorithmic}
\label{alg:path}
\end{algorithm}

\begin{lemma}
\label{lowerboundrpath}
The algorithm \emph{Path Augmentation( )} augments exactly $\lceil \: \frac{1}{2} \sum\limits_{i=1}^{r-1} (r-i) \times l_{i} \: \rceil$ edges.
\end{lemma}
\begin{proof}
Our claim is to prove that the algorithm augments exactly $\lceil \: \frac{1}{2} (r-1) \times l_{1} +  (r-2) \times l_{2}\: \rceil$ edges in the given path $P_n$. We shall prove this by varying $r$ and $n$ into odd and even. Our algorithm augments $\frac{l_1}{2}$ edges in \emph{Step 2}. Therefore, now there are $l_1+l_2$ degree two vertices, which is precisely the number of vertices in $G$. 
\begin{description}
\item[Case 1:] $r$ and $n$ are odd.\\
The \emph{Steps 3-7} augments $\frac{n+1}{2}$ edges. Since, the graph has only degree two vertices and the number of degree two vertices is $l_1+l_2$, $\frac{n+1}{2} = \frac{l_1+l_2+1}{2}$. In \emph{Steps 8-17}, $n$ edges are augmented for each $j \in \{2,\ldots, \lfloor \frac{r}{2} \rfloor \}$. Since $r$ is odd, the \emph{Steps 8-17} augments $\left( \frac{r-1}{2}-1\right)\cdot n$ edges, i.e., $\left( \frac{r-1}{2}-1\right)\cdot (l_1+l_2)$ edges. In total, the number of edges augmented by the algorithm is,
\begin{eqnarray}
\nonumber
& & \frac{l_1}{2} + \frac{l_1+l_2+1}{2} + \frac{(r-1)(l_1+l_2)}{2} - (l_1+l_2) \\ \nonumber
&=& \frac{l_1}{2} + \frac{l_1+l_2+1}{2} + \frac{(r-1)l_1}{2} + \frac{(r-1)l_2}{2} - (l_1+l_2) \\ \nonumber
&=& \frac{1}{2}\left[ l_1(r-1)+l_2(r-2)+1\right] \\ \nonumber
&=& \lceil \frac{1}{2}\{l_1(r-1)+l_2(r-2)\}  \rceil \text{ (Since, $r$ and $l_2$ are odd)}. 
\end{eqnarray}

\item[Case 2:] $r$ is odd and $n$ is even.\\
The \emph{Steps 3-7} augments $\frac{n}{2}$ edges i.e., $\frac{l_1+l_2}{2}$ edges. Similar to \emph{Case 1}, the \emph{Steps 8-17} augments $\left( \frac{r-1}{2}-1\right)\cdot n$ edges i.e., $\left( \frac{r-1}{2}-1\right)\cdot (l_1+l_2)$ edges. In total, the number of edges augmented by the algorithm is,
\begin{eqnarray}
\nonumber
& & \frac{l_1}{2} + \frac{l_1+l_2}{2} + \frac{(r-1)(l_1+l_2)}{2} - (l_1+l_2) \\ \nonumber
&=& \frac{l_1}{2} + \frac{l_1+l_2}{2} + \frac{(r-1)l_1}{2} + \frac{(r-1)l_2}{2} - (l_1+l_2)\\ \nonumber
&=& \frac{1}{2}\left[ l_1(r-1)+l_2(r-2)\right] \\ \nonumber
&=& \lceil \frac{1}{2}\{l_1(r-1)+l_2(r-2)\}  \rceil \text{ (Since, $r$ is odd and $l_2$ is even)}
\end{eqnarray}

\item[Case 3:] $r$ is even and $n$ is either odd or even.\\
Since $r$ is even, the \emph{Steps 8-17} augments $\left( \frac{r}{2}-1\right)\cdot n$ edges i.e., $\left( \frac{r}{2}-1\right)\cdot (l_1+l_2)$ edges. In total, the number of edges augmented by the algorithm is,
\begin{eqnarray}
\nonumber
& & \frac{l_1}{2}  + \frac{(r-2)(l_1+l_2}{2}\\ \nonumber
&=& \frac{l_1}{2} + \frac{(r-2)l_1}{2} + \frac{(r-2)l_2}{2}\\ \nonumber
&=& \frac{1}{2}\left[ l_1(r-1)+l_2(r-2)\right] \\ \nonumber
&=& \lceil \frac{1}{2}\{l_1(r-1)+l_2(r-2)\}  \rceil \text{ (Since, $r$ is even)}
\end{eqnarray}
Hence, the lemma is proved. $\hfill \qed$
\end{description}
\end{proof}

\begin{lemma}
\label{exactrpath}
 Let $P_n$ be a path on $n \geq 4$ vertices. Algorithm $\mathtt{Path~~ Augmentation()}$ yields a graph $G$, where $\forall~ v \in V(G), deg_G(v) \geq r$. Further, there exist at least one vertex $u \in V(G)$ such that $deg_G(v) = r$.
\end{lemma}
\begin{proof}
 The algorithm, first converts the path $P_n$ into a cycle $C_n$ by adding an edge between $v_1$ and $v_n$ in \emph{Step 2}. Now, the degree of each vertex in the resultant graph is two. If $r$ is odd, at the end of \emph{Steps 3-7}, the degree of each vertex is at least three \cite{dhana} and each iteration in \emph{Steps 8-17} increases the degree of all the vertices by two. The number of iterations in \emph{Steps 8-17} is $\frac{r-1}{2}-1$ and thus, the degree contribution to each vertex in $G$ in the \emph{Steps 8-17} is $\left(\frac{r-1}{2}-1\right) \times 2$. Thus, in the resulting graph, the degree of every vertex is at least $3 + \left(\frac{r-1}{2}-1\right) \times 2 = r$. If $r$ is even, each iteration in \emph{Steps 8-17} increases the degree of all the vertices by two and the number of iterations in \emph{Steps 8-17} is $\frac{r-1}{2}-1$. Thus, in the resulting graph, the degree of every vertex is $2 + \left(\frac{r}{2}-1\right) \times 2 = r$. Clearly, $deg_G(v_1) = r$. Hence, the lemma.
$\hfill \qed$
\end{proof}
\begin{definition}[Harary Graphs, $H_{r,n}$, $r<n$ \cite{dbwest}]
\label{defn:harary}
Place $n$ vertices around a circle equally spaced. If $r$ is even, form $H_{r,n}$ by making each vertex adjacent to the nearest $\frac{r}{2}$ vertices in each direction around the circle. If $r$ is odd and $n$ is even, form $H_{r,n}$ by making each vertex adjacent to the nearest $\frac{r-1}{2}$ vertices in each direction and diametrically opposite vertex. If $r$ and $n$ are odd, index the vertices by the integers $1$ to $n$. Construct $H_{r,n}$ from $H_{r-1,n}$ by adding edges $\{i,i+\frac{(n-1)}{2}\}$ for $1\leq i \leq \frac{(n+1)}{2}$.
\end{definition}

\begin{lemma}
\label{lemma:hararygraphs}
The graph $H$ obtained from the algorithm Path Augmentation() is a Harary graph.
\end{lemma}
\begin{proof}
If $r = n - 1$, then the algorithm augments edges between all the non-adjacent pair of vertices in $H$ and hence, $H$ is $r$-connected. Assume that $r < n - 1$. Since, the path $P_n$ is converted to a cycle in \emph{Step 2}, every vertex in $H$ is adjacent to a nearest vertex in each direction around the circle.
\begin{description}
\item[\textbf{Case 1:}] $r$ is even\\
 Every iteration $j\geq 2$ in the \emph{Steps 8-17}, makes a vertex $v_i$, $1\leq i \leq n$, in the graph $H$ adjacent to the nearest non-adjacent vertex in each direction around the circle. The iteration terminates when $j=\frac{r}{2}$ and hence, the number of iterations are $\frac{r}{2}-1$. Therefore, each vertex in $H$ is adjacent to the nearest $\frac{r}{2}$ vertices in each direction around the circle.
\item[\textbf{Case 2:}] $r$ is odd and $n$ is even\\
  Every iteration $j\geq 2$ in the \emph{Steps 8-17}, makes a vertex $v_i$, $1\leq i \leq n$, in the graph $H$ adjacent to the nearest non-adjacent vertex in each direction around the circle. The iteration terminates when $j=\frac{r-1}{2}$ and hence, the number of iterations are $\frac{r-1}{2}-1$. Therefore, each vertex in $H$ is adjacent to the nearest $\frac{r-1}{2}$ vertices in each direction and the iteration in \emph{Steps 3-7}, makes each vertex in $H$ adjacent to the diametrically opposite vertex.
\item[\textbf{Case 3:}] $r$ and $n$ are odd\\
The \emph{Steps 8-17}, constructs the $H_{k-1,n}$. Further, in \emph{Steps 3-7}, for every vertex $v_i$, $1 \leq i \leq \lceil \frac{n}{2} \rceil$, we augment an edge between $v_i$ and $v_{\lfloor \frac{n}{2}+i \rfloor} = v_{i+\frac{(n-1)}{2}}$.
\end{description}
From all the above cases and from the Definition \ref{defn:harary}, we conclude that the resultant graph $H$ is a Harary graph.
$\hfill \qed$
\end{proof}

\begin{lemma}
\label{connectivityhararygraphs}
Let $r<n$. Then $\kappa (H_{r,n}) = r$ \cite{harary}. 
\end{lemma}

\begin{lemma}
\label{pathrconnected}
For a path of length $n$, the graph obtained from the algorithm Path Augmentation() $H$ is $r$-connected.
\end{lemma}
\begin{proof}
It is clear from \emph{Lemma \ref{lemma:hararygraphs}} that the graph $H$ is a Harary graph. From \emph{Lemma \ref{connectivityhararygraphs}}, it therefore follows that $H$ is $r$-connected. Hence, the lemma.
$\hfill \qed$
\end{proof}

\clearpage

\subsection{Trace of Algorithm \ref{alg:path}}
\begin{figure}[H]
\centering
\includegraphics[scale=0.43]{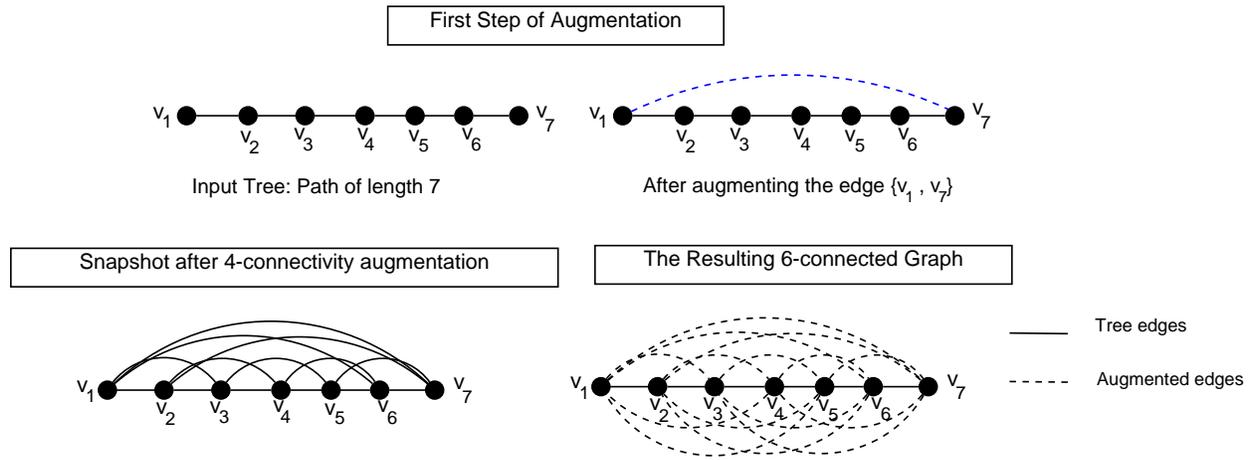}
\caption{Trace of Algorithm \ref{alg:path} when $n\,=\,7$ and $r\,=\,6$ }
\label{fig::tracepath2}
\end{figure}

\begin{figure}[H]
\centering
\includegraphics[scale=0.43]{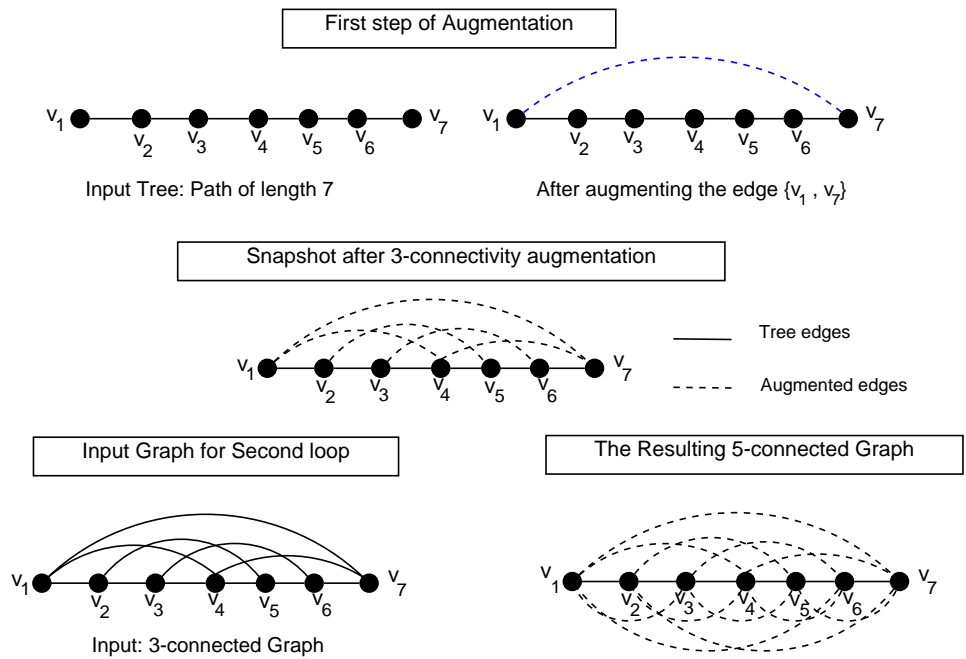}
\caption{Trace of Algorithm \ref{alg:path} when $n\,=\,7$ and $r\,=\,5$}
\label{fig::tracepath2}
\end{figure}


\subsection{$r$-connectivity augmentation for non-path trees}
In this subsection, we present an algorithm that finds a minimum set of edges whose augmentation to a tree makes it $r$-connected. 

\begin{definition}
\label{defn::c2}
Given a tree $T$, we define the $(C,2)$-\textbf{block tree} $T'$ as follows. Let $L(T) = \{l_1, l_2, \ldots, l_p\}$ be the set of leaves in $T$. Let $V(H) = V(T)$ and $E(H) = E(T) \cup \{\{l_i,l_{i+1}\} \vert 1 \leq i \leq p-1\}$. For each $1 \leq i \leq p-1$, the edge $\{l_i,l_{i+1}\}$ creates a fundamental cycle $C'$ in $H$. Let $C'$ and $C''$ be the fundamental cycles created due to $\{l_i,l_{i+1}\}, \{l_{i+1},l_{i+2}\}$, respectively. We refer to $C'$ and $C''$ as adjacent fundamental cycles. The $(C,2)$-\emph{block tree} $T'$ has the vertex set $V(T') = \{x \mid label(x) \text{ corresponds to an induced cycle of } H$ $\text{or a vertex separator of size 2 in } H$ $\text{or a vertex of degree 2 in } H\}$. For simplicity, we use the following notation. $V(T')$ consists of three kinds of vertices called $\sigma$ vertices, $\pi$ vertices and $\alpha$ vertices. We create a $\sigma$ vertex for a vertex separator of size two that separates adjacent fundamental cycles, a $\pi$ vertex for each fundamental cycle, a $\alpha$ vertex for each vertex of degree $2$. Note that, we do not create $\sigma$ vertex for each vertex separator of $H$, $\sigma$ vertex is created for a vertex separator $\{x,y\}$ such that either $\{x,y\}$ is fully contained in both $C$ and $C'$ or $x$ is contained in $C$ and $y$ is contained in $C'$, where $C$ and $C'$ are adjacent fundamental cycles. If $\alpha$ vertex is contained in both $C$ and $C'$, then it is adjacent to the $\pi$ vertex corresponding to $C$ and it is not adjacent to the $\pi$ vertex corresponding to $C'$. The adjacency between the pair of vertices in $V(T')$ is defined as follows: for $u,v \in V(T')$, $\{u,v\} \in E(T')$, if one of the following is true
\begin{itemize}
\item $u \in \sigma$ and $v \in \pi$ and $label(u) \subset label(v)$
\item $u \in \sigma$ and $v \in \pi$ and $label (u) = \{w,z\}$ such that $w \in label(v)$ and $z \in label (v')$ where $v$ and $v'$ are adjacent fundamental cycles in the graph $H$.
\item $u \in \pi$ and $v \in \alpha$ and $label(v) \subset label(u)$.
\end{itemize}
An illustration is given in \emph{Figure \ref{fig:c2block}}. 
\end{definition}

\begin{figure}[h]
\centering
\includegraphics[scale=0.3]{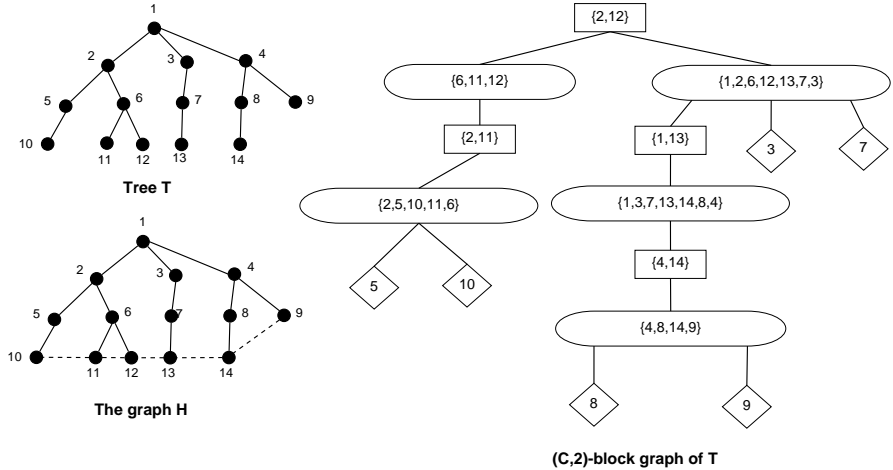}
\caption{A tree and the associated $(C,2)$-block tree $T'$}
\label{fig:c2block}
\end{figure}

\begin{definition}
Given a $(C,i)$-block tree $H$ of $G$, $i \geq 2$, we define $(C,i+1)$-\textbf{block tree} $T$ as follows; the structure of $T$ remains the same and the only change is the set of $\alpha$ vertices. $\alpha$ vertices in $(C,i+1)$-block tree are vertices of degree $i+1$. Note that $\alpha$ vertices in $(C,i)$ becomes vertices of degree $i+1$ after edge augmentation. The other vertices of degree $i+1$ in $C$ are also included as $\alpha$ vertices in $(C,i+1)$ (refer \emph{Figure \ref{fig::tracenonpath2}}, a $(C,3)$-block tree). The \textbf{degree} of a $\pi$ vertex is the number of $\alpha$ vertices adjacent to it in $T$ and it is denoted by $deg(\pi)$.
\end{definition}

\begin{figure}[h]
\centering
\includegraphics[scale=0.23]{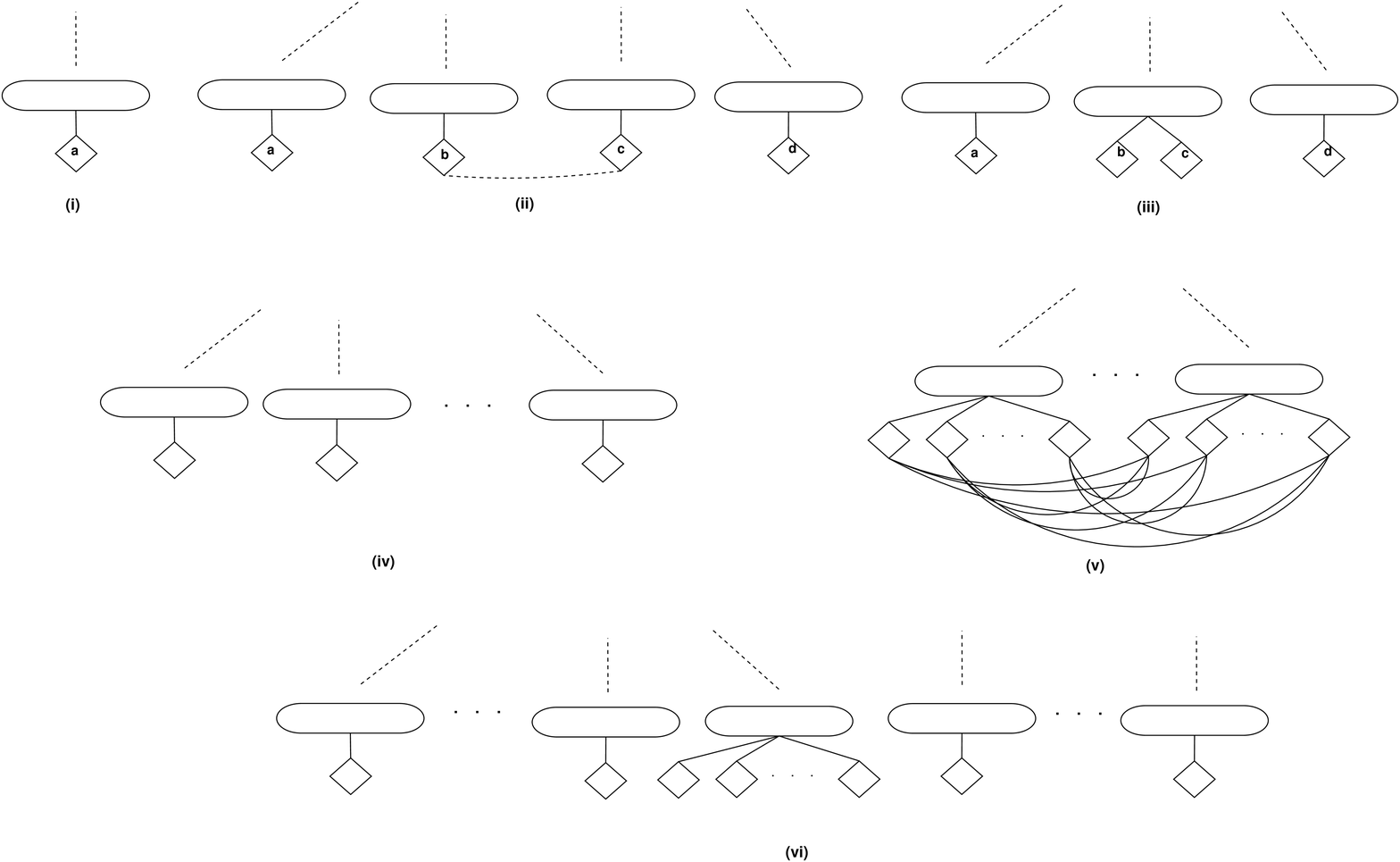}
\caption{(i) Boundary Case 1, (ii) Boundary Case 2, (iii) Boundary Case 3, (iv) Boundary Case 4, and (v) Boundary Case 5}
\label{fig::bc}
\end{figure}
%
%
%
%
%
%

\noindent \textbf{Outline of the algorithm:} Given a tree $T$, construct $(C,2)$-block tree as per \emph{Definition \ref{defn::c2}}, and we consider only the $\pi$ vertices have degree at least one. We augment edges as follows: identify two $\pi$ vertices $\pi_{max}$ and $\pi_{smax}$ such that the corresponding degree in $T'$ is maximum and second maximum, add the edge $\{x,y\}$, $x$ is a $\alpha$ vertex adjacent to $\pi_{max}$ vertex and $y$ is a $\alpha$ vertex adjacent to $\pi_{smax}$ vertex. Remove $x$ and $y$ from $T$ to get a tree again for the next iteration. This process is stopped when we encounter one of the following boundary cases for which augmentation is done separately.
\begin{description}
\item[\textbf{Boundary Case 1:}] $T'$ has exactly one $\alpha$ vertex (\emph{see Figure \ref{fig::bc} (i)}), say $a$, and let the corresponding $\pi$ vertex be $\pi_i$. In this case, augment the edge $\{a,c\}$, where $c \in label(\pi_j)$, $i\neq j$, and $deg_{T'}(c)=deg_{T'}(a)+1$.

\item[\textbf{Boundary Case 2:}] $T'$ has exactly four $\pi$ vertices, say $\pi_1, \pi_2, \pi_3, \pi_4$, and each $\pi$ vertex has exactly one $\alpha$ vertex in it, say $a, b, c$ and $d$, respectively. Also, the edge $\{b,c\}$ exist (\emph{see Figure \ref{fig::bc} (ii)}). In the case where $\{b,c\}$, augment the edges $\{a,c\}$ and $\{b,d\}$. 

\item[\textbf{Boundary Case 3:}]  $T'$ has exactly three $\pi$ vertices with degree sequence $(1,2,1)$ as illustrated in \emph{Figure \ref{fig::bc} (iii)}. In this case, augment the edges $\{a,c\}$ and $\{b,d\}$. If either the edge $\{a,c\}$ or the edge $\{b,d\}$ exist, then augment the edges $\{a,b\}$ and $\{c,d\}$.

\item[\textbf{Boundary Case 4:}] The degree of each $\pi$ vertex in $T'$ is one. For each $1 \leq i \leq s$, augment an edge between a $\alpha$ vertex of $\pi_i$ and a $\alpha$ vertex of $\pi_j$ such that $\vert i-j \vert$ is maximum (\emph{see Figure \ref{fig::bc} (iv)}). After all such augmentation is done, the resultant graph may belong to either boundary case 1 or boundary case 2, in which case augmentation is done as per \emph{boundary Cases 1} and $2$.

\item[\textbf{Boundary Case 5:}] The degree of all $\pi$ vertices is at least one and each $\alpha$ vertex of $\pi_i$, $1\leq i \leq s$, is adjacent to all $\alpha$ vertices of $\pi_j$, $1\leq j \leq s$ and $i \neq j$ (\emph{see Figure \ref{fig::bc} (v)}). Augment edges between two non-adjacent $\alpha$ vertices of the same $\pi$ vertex. After augmentation remove $x$ and $y$ to get a tree for the next iteration. After all such augmentation is done, the resultant graph may belong to boundary case 1, in which case augmentation is done as per \emph{boundary Case 1}.

\item[\textbf{Boundary Case 6:}] The degree of $\pi_1, \ldots, \pi_{l-1}, \pi_{l+1}, \ldots, \pi_s$ vertices are exactly one and $deg(\pi_l) > 1$, $1 \leq l \leq s$ (\emph{see Figure \ref{fig::bc} (vi)}). Augment the edges between a $\alpha$ vertex adjacent to $\pi_l$ and a $\alpha$ vertex adjacent to $\pi_i$ vertex, $i \in \{1,\ldots,l-1,l+1,\ldots,s\}$, such that $\vert i-l \vert$ is maximum. Note that, whenever we augment an edge between two $\alpha$ vertices $x$ and $y$, we remove $x$ and $y$ from $\alpha$ vertices. After all such augmentation is done, the resultant graph may belong to boundary case 1 or boundary case 2 or boundary case 3, in which case augmentation is done as per \emph{boundary Cases 1, 2 and 3}.
\end{description}

At the end of the first iteration, the given tree $T$ is made 3-connected with the help of the associated $(C,2)$-block tree $T'$. We next construct $(C,3)$-block tree as per the definition and use it to get a 4-connected graph. We repeat this process till augmentation for $(C,r-1)$-block tree is done. This completes the algorithm.

\begin{algorithm}[H]
\caption{\tt $r$-connectivity Augmentation in non-path trees: $Non-Path \: Augmentation\,(Tree\:T)$}
\begin{algorithmic}[1]
\STATE{\textbf{Input:} A tree $T$ with $n$ vertices and an integer $r < n$.}
\STATE{\textbf{Output:} $r$-connected graph $H$ of $T$.}
\STATE{Let $x_1,\ldots, x_p$ be the leaves of $T$. Thus, $E_{ca} = \{\{x_i,x_{i+1}\} \vert 1 \leq i \leq p-1 \}$.}
\FOR{$j=2$ to $r-1$}
\STATE{ Compute the $(C,j)$-block tree $T'$ and let the number of $\pi$ vertices in $T'$ be $s$.}
\WHILE{$T'$ does not belongs to any boundary case $k$, $1 \leq k \leq 6$}
\STATE{Let $\pi_{max}$ and $\pi_{smax}$ are the two $\pi$ vertices in $T'$ such that $deg(\pi_{max}) \geq deg(\pi_{smax}) \geq deg(\pi_i), 1 \leq i \leq p-1$ and $i \neq max \neq smax$.}
\STATE{Add the edge $\{x,y\}$, $x$ is a $\alpha$ vertex adjacent to $\pi_{max}$, $y$ is a $\alpha$ vertex adjacent to $\pi_{smax}$. Remove the $\alpha$ vertices $x$ and $y$ from $T'$. /* \tt{This step yields a tree for the next iteration} */}
\STATE{Add $\{x,y\}$ to $E_{ca}$}
\STATE{Update $deg(\pi_{max})$ and $deg(\pi_{smax})$.}
\ENDWHILE
\IF{$T'$ belongs to boundary case $k$, $1 \leq k \leq 6$}
\STATE{Augment edges between $\alpha$ vertices as mentioned.}
\ENDIF
\ENDFOR
\STATE{Return $E_{ca}$ and the resultant graph $H$}
\end{algorithmic}
\label{alg:nonpath}
\end{algorithm}

\begin{lemma}
\label{lowerboundnon-path}
The algorithm \emph{Non-Path Augmentation( )} augments exactly $\lceil \: \frac{1}{2} \sum\limits_{i=1}^{r-1} (r-i) \times l_{i} \: \rceil$ edges.
\end{lemma}
\begin{proof}
Algorithm \ref{alg:nonpath} augments $p-1$ edges to $T$ during the computation of $(C,2)$-block tree. Later, at every iteration, the algorithm augments exactly the half the number of degree $i$, $1 \leq i \leq r-1$, vertices in the resultant graph of the previous iteration. i.e., degree $i$ vertices are converted into degree $r$ vertices in $T$, $1 \leq i \leq r-1$, by augmenting $ \lceil \: \frac{1}{2} \sum\limits_{i=1}^{r-1} (r-i) \times l_{i} \: \rceil  $ edges. This proves the lemma. $\hfill \qed$
\end{proof}

\begin{lemma}
\label{exactrnonpath}
 Let $T$ be a non-path tree on $n \geq 4$ vertices. Algorithm $\mathtt{Path~~ Augmentation()}$ yields a graph $G$, where $\forall~ v \in V(G), deg_G(v) \geq r$. Further, there exist at least one vertex $u \in V(G)$ such that $deg_G(v) = r$.
\end{lemma}
\begin{proof}
The proof trivially follows from the algorithm.$\hfill \qed$

\end{proof}

\begin{lemma}
\label{lemma::bc1}
Let $T'$ be the $(C,i)$-block tree which is either boundary case 1 or boundary case 2. Then, the output graph $H$ is $(i+1)$-connected.
\end{lemma}
\begin{proof}
 It is easy to see that there is no separator of size $i$ in $H$ as per the strategy followed for augmentation.$\hfill \qed$
\end{proof}

\begin{lemma}
\label{lemma::bc3}
Let $T'$ be the $(C,2)$-block tree satisfying boundary case 3. Then, the output graph $H$ is $3$-connected.
\end{lemma}
\begin{proof}
We shall prove this lemma by mathematical induction on $s$, number of $\pi$ vertices.\\
\textit{Base Case:} $s=2$, say $\pi_1$ and $\pi_2$. Let $a$ and $b$ be the $\alpha$ vertices of $\pi_1$ and $\pi_2$, respectively. By augmenting the edge $\{a,b\}$, the graph $H$ has no 2-size separator. Hence, it is 3-connected.\\
\textit{Hypothesis:} Assume that the lemma is true for $s \geq 2$.\\
\textit{Induction Step:} 
Let $\pi_1, \ldots, \pi_s$ be the $\pi$ vertices of $T'$. It is given that $deg(\pi_i)=1$, $1 \leq i \leq s$. Our algorithm augments an edge between the $\alpha$ vertex of $\pi_1$, say $a$, and the $\alpha$ vertex of $\pi_k$, say $b$, where $b$ is the largest non-adjacent $\alpha$ vertex of $a$. Let $T''$ be the tree $T'\backslash \{a,b\}$. i.e., the number of $\pi$ vertices in $T''$ is $s-2$. By the hypothesis, our approach yields a 3-connected graph of $T''$. Now, the only separators of size two in $T'$ are the neighbors of $a$ or $b$, which would be removed by the augmented edge   $\{a,b\}$. Thus, $H$ is 3-connected.$\hfill \qed$
\end{proof}

\begin{lemma}
\label{lemma::bc3itoi1}
Let $T'$ be the $(C,i)$-block tree satisfying boundary case 3. Then, $H$ is $(i+1)$-connected.
\end{lemma}
\begin{proof}
Let us prove this lemma by mathematical induction on $i$.\\
\textit{Base Case:} The lemma is true for $i=2$ by the \emph{Lemma \ref{lemma::bc3}}.\\
\textit{Hypothesis:} Assume that the lemma is true for $i \geq 2$.\\
\textit{Induction Step:} 
Let $\pi_1, \ldots, \pi_s$ be the $\pi$ vertices of $T'$. It is given that $deg(\pi_i)=1$, $1 \leq i \leq s$. Our algorithm augments an edge between the $\alpha$ vertex of $\pi_1$, say $a$, and the $\alpha$ vertex of $\pi_k$, say $b$, where $b$ is the largest non-adjacent $\alpha$ vertex of $a$. Let $T''$ be the tree $T'\backslash \{a,b\}$. i.e., the number of $\pi$ vertices in $T'$ is $s-2$. By the hypothesis, graph obtained from our approach is $(i+1)$-connected. Now, the only separators of size $i$ in $T$ are the set of vertices where each vertex lies in the distinct path between $a$ and $b$. We remove all such separators by augmenting an edge $\{a,b\}$. Thus, $H$ is $(i+1)$-connected.$\hfill \qed$
\end{proof}

\begin{lemma}
\label{lemma::bc4itoi1}
Let $T$ be the $(C,i)$-block tree satisfying boundary case 4. The graph $H$ obtained from our algorithm is $(i+1)$-connected.
\end{lemma}
\begin{proof}
Let $T'$ be the $(C,i)$-block tree satisfying boundary case 4. By our algorithm each $\pi$ vertex is converted into a Harary graph, with connectivity $i+1$. Since all possible edges exist between any two $\pi$ vertices, the resulting graph $H$ is $(i+1)$-connected. $\hfill \qed$
\end{proof}

\begin{lemma}
\label{lemma::bc5itoi1}
Let $T'$ be the $(C,i)$-block tree satisfying boundary case 5. Then $H$  is $(i+1)$-connected.
\end{lemma}
\begin{proof}
We shall prove this lemma by mathematical induction on $s$, the number of $\pi$ vertices.\\
\textit{Base Case:} $s=1$. Our approach yields a Harary graph $H_{i+1,n}$ when $s=1$. By \emph{Lemma \ref{connectivityhararygraphs}}, $H$ is $(i+1)$-connected.\\
\textit{Hypothesis:} Assume that the lemma is true for $s \geq 2$.\\
\textit{Induction Step:}
Let $\pi_1, \ldots, \pi_s$ be the $\pi$ vertices of $T'$. Let $T''$ be the tree $T'\backslash \{\pi_s \cup \{\alpha \text{ vertices of }\pi_s\}\}$. By the hypothesis, $H$ is $(i+1)$-connected. Now, the only non-adjacent vertices with separator of size $i$ in $T'$ are the $\alpha$ vertices of $\pi_s$. We augment edges between them and this removes all separators of size $i$. Hence, the lemma. $\hfill \qed$
\end{proof}

\begin{lemma}
\label{nonpathr}
For a non-path tree $T$, the graph $H$ obtained from the algorithm $\mathtt{Non}$-$\mathtt{Path~~Augmentation()}$ is $r$-connected.
\end{lemma}
\begin{proof}
Let us prove this lemma by mathematical induction on $k$, the number of $\alpha$ vertices in $T'$. If $T'$ satisfies any of the base cases, then from the lemmas presented above, it is clear that $H$ is $r$-connected. Assume $T'$ is a non-trivial block tree. By our algorithm, we choose an $\alpha$ vertex $x$ from $\pi_{max}$ and an $\alpha$ vertex $y$ from $\pi_{smax}$ such that $\{x,y\} \notin E_{ca}$. We augment an edge $\{x,y\}$, remove $x$ and $y$ from $T'$ to get a tree $T''$. By the hypothesis, $T''$ has fewer $\alpha$ vertices than $T'$ and hence our approach guarantees a $r$-connected graph $H'$. Clearly, the edge $\{x,y\}$ when introduced to $H'$ removes vertex separators and thus, $H$ is $r$-connected. $\hfill \qed$
\end{proof}

\begin{theorem}
\label{finalr}
For a tree $T$, the graph $H$ obtained from Algorithm \ref{alg:tree} is $r$-connected. Further, $H$ is obtained from $T$ by augmenting a minimum set of edges.
\end{theorem}
\begin{proof}
The lower bound for the $r$-connectivity augmentation of $T$ is $ \lceil\frac{1}{2} \sum\limits_{i=1}^{r-1} (r-i) \times l_{i} \rceil $, where $l_i$ denotes the number of vertices of degree $i$ in $T$ by \emph{Lemma \ref{lb}}. The algorithm calls the Algorithm \ref{alg:path} if $T$ is a path. The Algorithm \ref{alg:path} converts the path to a $r$-connected graph $H$ by augmenting exactly $ \lceil\frac{1}{2} \sum\limits_{i=1}^{r-1} (r-i) \times l_{i} \rceil $ edges (by \emph{Lemma \ref{lowerboundrpath}} and \emph{Lemma \ref{pathrconnected}}). The algorithm calls the Algorithm \ref{alg:nonpath} if $T$ is a non-path. The Algorithm \ref{alg:nonpath} converts the non-path tree $T$ to a $r$-connected graph $H$ by augmenting exactly $ \lceil\frac{1}{2} \sum\limits_{i=1}^{r-1} (r-i) \times l_{i} \rceil $ edges (by \emph{Lemma \ref{lowerboundnon-path}} and \emph{Lemma \ref{nonpathr}}). Thus, for a tree $T$ the graph obtained from Algorithm \ref{alg:tree} is $r$-connected. Further, $H$ is obtained by using a minimum connectivity augmentation set. Therefore, the claim follows. $\hfill \qed$
\end{proof}

%

\subsection{Trace of Algorithm \ref{alg:nonpath}}
We trace algorithm \ref{alg:nonpath} for the tree given in \emph{Figure \ref{fig:c2block}} and for $r=4$. The first step of augmentation among the leaves is explained in \emph{Figure \ref{fig:c2block}}. The rest of the iterations are as follows.

\begin{figure}[H]
\centering
\includegraphics[scale=0.27]{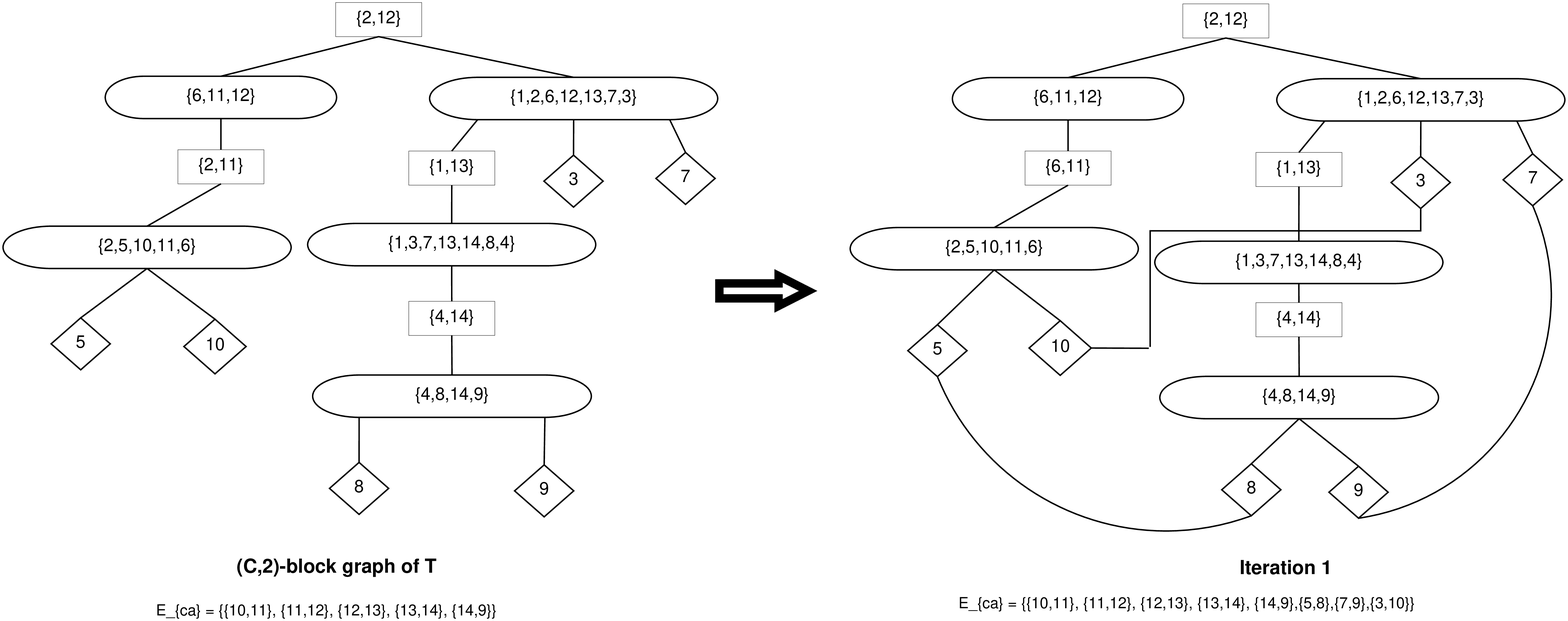}
\label{fig::tracenonpath1}
\end{figure}

\begin{figure}[H]
\centering
\includegraphics[scale=0.27]{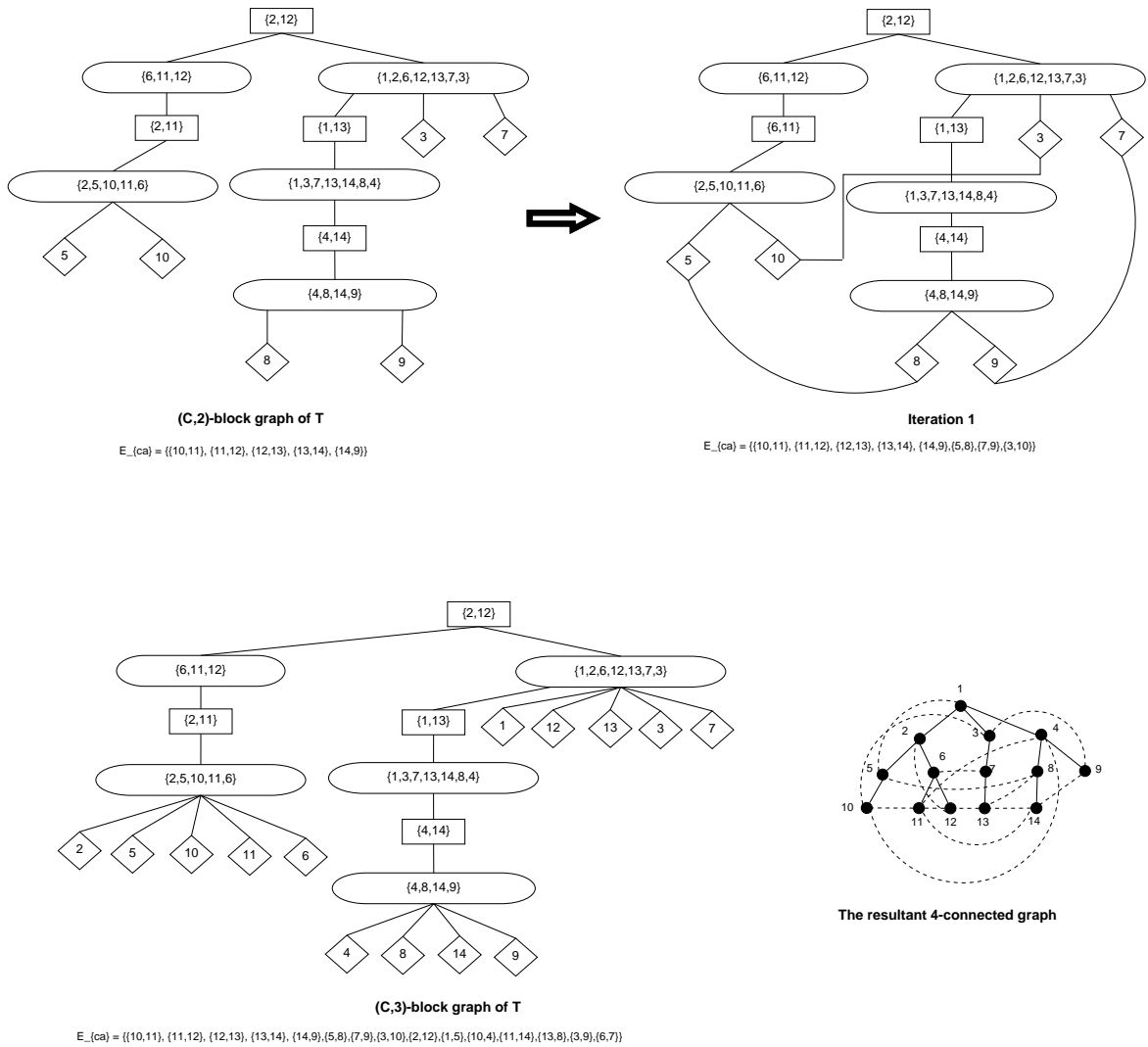}
\caption{An illustration for Algorithm \ref{alg:nonpath}}
\label{fig::tracenonpath2}
\end{figure}

\section{Conclusion and Future work}
In this paper, we have presented an algorithm for finding a minimum $r$-connectivity augmentation set in trees. We believe that this approach can be extended to 1-connected graphs. An extension of this work would be the $(k+r)$-connectivity augmentation of $k$-connected graphs, for any $k < r < n$.

\end{document}